\documentclass{article}

\usepackage[
    persons={Alireza,Nima},
    bibliosources={refs.bib}
]{pomegranate}
\usepackage{algorithm}
\usepackage{algpseudocode}

\DeclareDelimiter{\Otilde}[\mathnormal{\widetilde{O}}]{\lparen}{\rparen}
\DeclareDelimiter{\Reff}[{\mathcal{R}_{\textnormal{eff}}}]{\lparen}{\rparen}

\DeclareDelimiter{\dTV}[\mathnormal{d}_{\operatorname{TV}}]{\lparen}{\rparen}

\title{Fast Spanning Tree Sampling in Broadcast Congested Clique}
\author[1]{Nima Anari}
\author[1]{Alireza Haqi}
\affil[1]{Stanford University, \url{{anari,ahaqi}@stanford.edu}}
\date{}

\begin{document}

    \maketitle
    \begin{abstract}
        We present the first polylogarithmic-round algorithm for sampling a random spanning tree in the (Broadcast) Congested Clique model.
For any constant $c > 0$, our algorithm outputs a sample from a distribution whose total variation distance from the uniform spanning tree distribution is at most $O(n^{-c})$ in at most $c \cdot \log^{O(1)}(n)$ rounds. The exponent hidden in $\log^{O(1)}(n)$ is an absolute constant independent of $c$ and $n$. This is an exponential improvement over the previous best algorithm of Pemmaraju, Roy, and Sobel (PODC 2025) for the Congested Clique model.

    \end{abstract}
    
    \section{Introduction}\label{sec:intro}

Random spanning trees are a central object of study in combinatorics, probability, and algorithms.
The problem of sampling a uniform spanning tree goes back to the classical work of Kirchhoff and Cayley \cite{kirby2016kirchhoff}.
More generally, given edge weights $\set{w_e\given e\text{ is an edge of the graph}}$, the weighted spanning tree distribution is defined by
\[ \P{T} \propto \prod_{e\in T} w_e. \]

In algorithmic applications, uniform and weighted spanning tree distributions have been used to analyze graph sparsifiers \cite{kyng2018matrix, goyal2009expanders}, design approximation algorithms for the Traveling Salesperson Problem (TSP) \cite{genova2017experimental, karlin2021slightly}, and study the minimum $k$-edge-connected multisubgraph problem \cite{karlin2022improved}.

These applications motivated a long sequence of works on fast sampling algorithms.
The algorithmic study of this problem has a long history: \textcite{broder1989generating,aldous1990random} gave an elegant random-walk-based algorithm.
A long line of work based on the graph Laplacian solving paradigms \cite{wilson1996generating, colbourn1996two, kelner2009faster, madry2014fast, durfee2017sampling, durfee2017determinant} culminated in almost linear time sampling in centralized models of computation \cite{schild2018almost}; more recent techniques based on matroid exchange walks yielded improved nearly linear time algorithms \cite{anari2021log}.

Recently, \textcite{PRS25} initiated the study of sampling spanning trees in distributed models.
They gave an algorithm with runtime $\Otilde{n^{1/2 + \alpha}}$, where $\alpha$ denotes the runtime exponent for matrix multiplication in the Congested Clique model; their approach also has an inherent $\Omega(\sqrt{n})$ bottleneck.
In this paper, we settle the complexity of the problem up to polylogarithmic factors in two distributed models, Congested Clique and Broadcast Congested Clique, by giving an $\Otilde{1}$ round algorithm.

A standard framework for distributed graph algorithms is the Congested Clique model.
In this model, computation proceeds in synchronous rounds over a complete communication network: each node initially knows only its incident edges in the input graph, and in each round every pair of nodes can exchange messages of size $O(\log n)$ bits.
A wide variety of classic problems have been studied in this model, including $(\Delta + 1)$ coloring \cite{czumaj2020simple}, minimum spanning tree \cite{ghaffari2016mst,lotker2003mst, nowicki2021deterministic}, maximal independent set (MIS) \cite{ghaffari2017distributed}, shortest paths \cite{nanongkai2014distributed,becker2021near, dory2022exponentially, dory2021constant}, and spanners \cite{parter2018congested, dory2021constant}.

A more restrictive variant is the Broadcast Congested Clique (BCC), formalized and first introduced in \cite{drucker2014power,ghaffari2015near}, where each node in each round can only broadcast one identical message to all other nodes.
This restriction models scenarios in which individualized communication is infeasible or expensive, while uniform broadcast is cheap.
Although the BCC model is weaker and algorithm design is often more challenging, many problems still admit efficient algorithms \cite{drucker2014power, holzer2014approximation,becker2016effect, jurdzinski2018connectivity}.
Alongside these positive results, a long line of work establishes lower bounds in the BCC model \cite{czumaj2020detecting, drucker2014power, becker2016effect, korhonen2017deterministic}.

\subsection{Our Contribution}
In this work, we study the problem of sampling spanning trees in the Congested Clique and Broadcast Congested Clique models.
To our knowledge, we give the first polylogarithmic round BCC algorithm for sampling from weighted spanning tree distributions.
This is an exponential improvement over the previous best Congested Clique result of \cite{PRS25}, which runs in $\Otilde{n^{1/2 + \alpha}}$, where $\alpha$ is the runtime exponent of matrix multiplication in the Congested Clique model (currently $\alpha = 1 - 2/w$, where $w$ is the sequential matrix multiplication exponent).

\begin{theorem}[\textbf{Main Result}] \label{thm:main-result}
There exists a universal constant $C$ and an algorithm $\mathcal{A}$ such that for every graph $G=(V,E,w)$ with $w \in \mathbb{Z}_{\ge 0}^{|E|}$ and $\|w\|_\infty \le U$, $\mathcal{A}$ outputs a sample from a distribution $\mu'$ over spanning trees satisfying
\[
    \dTV{\mu, \mu'} \le \varepsilon.
\]

Furthermore, $\mathcal{A}$ runs in
\[
\log\!\left(\frac{1}{\varepsilon}\right)\log(nU)\log^C(n)
\]
Broadcast Congested Clique rounds, which implies the same guarantee in the Congested Clique model.
\end{theorem}

\subsection{Related Work}

The literature on distributed computing, including Congested Clique, Broadcast Congested Clique, LOCAL, CONGEST, and Massively Parallel Computation (MPC), has mostly focused on decision and search problems.
Sampling problems are comparatively less explored.
Classic examples studied through the distributed lens include vertex coloring \cite{feng2018distributed, barenboim2011deterministic, ghaffari2022deterministic, barenboim2017deterministic}, CSP \cite{bordeaux2009experiments, chung2014distributed}, maximal independent set \cite{ghaffari2018improved, faour2025faster, efron2020beyond}, and minimum spanning tree \cite{ghaffari2018improved, nowicki2021deterministic, hegeman2015toward}.
A smaller subset of these works studies sampling variants \cite{carlson2023improved, fischer2018simple, feng2017can, feng2018distributed}.
Our work contributes to this sampling perspective by giving an exponential improvement for random spanning tree sampling and the first polylogarithmic time algorithm in the Broadcast Congested Clique model.

The efficiency of CONGEST algorithms is measured by the number of synchronous rounds until termination.
Ideally, one seeks $O(1)$ round algorithms, but some tasks (for example, diameter computation) require at least as many rounds as the graph diameter \cite[see][]{Frischknecht2012, Peleg2000, Awerbuch1985, KuttenPeleg1998, Kuhn2006}.
Beyond dependence on $n$, the literature also studies bounds in terms of structural graph parameters such as diameter, arboricity \cite{morgan2021algorithms}, and maximum degree; sublinear dependence on these parameters is often desirable. The maximum degree parameter also plays an important role in the analysis of many distributed graph problems, including maximal independent set \cite{kawarabayashi2019improved}. Some problems are also analyzed through problem specific quantities, such as random walk mixing time and conductance \cite{molla2017distributed}.

In the Congested Clique model, the communication network is a clique (not necessarily the input graph itself). This gives communication diameter $1$, but it does not force structural parameters of the input graph, such as maximum degree, to be constant.
As in CONGEST, the best possible complexity is still $O(1)$ rounds.
Several fundamental problems admit constant round algorithms in Congested Clique, including minimum spanning tree \cite{nowicki2021deterministic}, graph coloring \cite{czumaj2020simple}, and graph spanners \cite{dory2021constant}.

The problem of distributed spanning tree sampling was first studied by \textcite{PRS25}, who gave a sublinear-time algorithm in the Congested Clique model running in $\Otilde{n^{1/2+\alpha}}$ rounds. Here, $\alpha$ denotes the matrix-multiplication exponent in the Congested Clique model (currently $\alpha = 1 - \frac{2}{\omega}$, where $\omega$ is the sequential matrix-multiplication exponent). Their algorithm builds on the classic Aldous--Broder random-walk method, where a transcript of a random walk on the graph is translated to a random spanning tree by keeping the edges where the walk first visits each vertex. \textcite{PRS25} develop a bottom-up framework for constructing (a reduced transcript of) a long random walk via shortcutting and Schur complements. This gives polylogarithmic runtime when the graph has near-linear cover time; in general, however, the runtime is $\Otilde{n^{1/2+\alpha}}$, with an inherent $\Omega(\sqrt{n})$ bottleneck, regardless of improvements in the matrix-multiplication exponent.

We take a different approach. We base our distributed algorithm on the approximate spanning tree sampling framework of \textcite{anari2024optimal}. In this framework, sampling a spanning tree from a dense graph is reduced to a series of similar sampling problems involving sparsified graphs. We show how to compute the sparsified graphs in $\Otilde{1}$ rounds and to collect the edges in one machine, where we can trivially perform the subtask of sampling a spanning tree. A key ingredient in this algorithm is estimating the marginal probability that each edge belongs to a random spanning tree. To compute these quantities, we use the distributed algorithm of \textcite{forster2022laplacian}, which finds Laplacian sparsifiers based on probabilistic spanner techniques \cite[see][]{becker2021near, baswana2007simple}.

In this paper, our main result is a polylogarithmic-round algorithm in the Broadcast Congested Clique model for sampling random spanning trees by instantiating the approximate sampling algorithm of \textcite{anari2024optimal}.

\subsection{Acknowledgements}
This work was supported by NSF CCF-2045354.

    \section{Preliminaries}\label{sec:prelims}

We introduce notation and standard results used throughout the paper.
We begin with the Congested Clique and Broadcast Congested Clique models, and then review the ingredients needed for our sampling algorithm: the complement down-up walk and polylogarithmic-round Laplacian tools from \cite{forster2022laplacian}.

Throughout the paper, we use $\Otilde{1}$ and $\poly\log(n)$ interchangeably.

\subsection{Congested Clique and Broadcast Congested Clique}
\begin{definition}[CONGEST model] \label{def:Congested-model}
Given a graph $G = (V, E)$ with $n = \card{V}$ vertices, the CONGEST model is defined as follows:
\begin{itemize}
    \item There are $n$ machines $M_1, M_2, \dots, M_n$ with a bijection between machines and graph vertices.
    \item For every edge $(u,v) \in E$, the corresponding machines are connected by a dedicated communication channel of bandwidth $O(\log n)$ bits per round. No communication links exist between non-adjacent vertices.
    \item Each node initially knows its adjacent nodes and incident communication edges.
\end{itemize}
\end{definition}

For convenience, we introduce notation for incident edges.
\begin{definition}
Given a graph $G = (V, E)$ and a vertex $v \in V$, let $N_v^G$ denote the set of edges incident to $v$.
\end{definition}

The Congested Clique is a special case of CONGEST (\cref{def:Congested-model}) in which the communication network is a clique.
\begin{definition}[Congested Clique model] \label{def:Congested-Clique-model}
The Congested Clique model is the variant of CONGEST in which, for an input graph $G=(V,E)$, each machine knows the neighbors of its corresponding vertex in $G$ and can communicate with every other machine via a dedicated $O(\log n)$-bit channel per round.
\end{definition}

Since assuming $\Theta(n^2)$ communication links is strong, a natural restriction is the Broadcast Congested Clique, where each machine must send the same message to all others in a round.

\begin{definition}[Broadcast Congested Clique model]\label{def:Broadcast-Congested-Clique}
Broadcast Congested Clique is the restricted version of Congested Clique (\cref{def:Congested-Clique-model}) in which each node can only broadcast one identical message to all nodes in each round.
\end{definition}

In all three models above, every node $v\in V$ knows its incident edge set $N_v^G$ in the input graph $G$.

Algorithms in CONGEST, Congested Clique, and Broadcast Congested Clique execute in synchronous rounds.
Depending on the model, a machine may communicate only with graph neighbors or with all machines; in the broadcast model, that message must be identical for all recipients.
An algorithm terminates when a valid solution is produced and all machines know that the computation has terminated.

\begin{lemma}\label{lemma:initial-spanning-tree}
Given a connected input graph $G$, one can compute a spanning tree in $\Otilde{1}$ rounds in the Broadcast Congested Clique model \cite{jurdzinski2017brief}.
\end{lemma}

\subsubsection{Data Transmissions}

Our algorithm repeatedly gathers sampled edges to a single machine.
The sampled subgraphs are sparse in a density sense, which allows efficient collection in BCC.
We first recall subset density.

\begin{definition}
For a graph $G=(V,E)$ and nonempty $S\subseteq V$, define the subset density of $S$ as
\[
\frac{|E(S)|}{|S|},
\]
where $E(S)$ is the set of edges induced by $S$.
Define the graph density
\[
\rho(G) = \max_{\emptyset \neq S\subseteq V} \frac{|E(S)|}{|S|}.
\]
For a weighted graph $G=(V,E,w)$, define weighted densities analogously by replacing $|E(S)|$ with $w(E(S)) = \sum_{e\in E(S)} w_e$.
\end{definition}

Trees and forests have density at most one.
By a standard observation, the weighted graph whose edge weights are effective resistances also has subset density at most one.
We exploit this property to associate each edge to one of its endpoints in the BCC model so that each vertex receives total associated weight $\Otilde{1}$; we call this a light association.

\begin{definition}\label{def:light-association}
Given a weighted graph $G=(V,E)$ with edge weights $q_e\ge 0$, a collection $(N_v^o)_{v\in V}$ is a \emph{light edge association} if
\begin{itemize}
    \item for every $v\in V$, $N_v^o \subseteq N_v$,
    \item $\bigcup_{v\in V} N_v^o = E$ and the sets are pairwise disjoint, and
    \item for every $v\in V$, $\sum_{e\in N_v^o} q_e \le \Otilde{1}$.
\end{itemize}
\end{definition}

Equivalently, a light edge association is an edge orientation under a fixed convention: if $e=\{u,v\}$ is associated to $u$, view $e$ as oriented from $v$ to $u$.
Thus, throughout the paper, ``orientation'' and ``association'' both mean assigning each edge to one of its endpoints.

After computing a light association and sampling edges, we can gather sampled edges at one machine in $\Otilde{1}$ rounds, because each vertex contributes only $\Otilde{1}$ expected/typical sampled edges. We formalize this in \cref{prop:bounded_degree_edge_proof}.
\begin{proposition}[Informal; see \cref{prop:bounded_degree_edge_proof}]
Given overestimates $q_e$ of spanning-tree edge marginals (equivalently, overestimates of $w_e\Reff{e}$), we can compute a light association in $\Otilde{1}$ BCC rounds.
\end{proposition}

\subsection{Complement Down-Up Walk}
To design our BCC sampler, we simulate the down-up walk framework of \textcite{anari2024optimal}.
That work proves near-optimal mixing for certain down-up walks on strongly Rayleigh distributions, under isotropy assumptions.

We use standard definitions from that literature.

\begin{definition}(Restricted distribution)
\label{def:restricted_distribution}
For a distribution $\mu$ over subsets of $[n]$ and $S\subseteq [n]$, define $\mu_S$ as the distribution of $F\sim\mu$ conditioned on $F\subseteq S$.
\end{definition}

\begin{definition}(Conditional distribution)
\label{def:conditional-distribution}
For a distribution $\mu$ over subsets of $[n]$ and $T\subseteq [n]$, define $\mu^T$ as the distribution of $F\sim\mu$ conditioned on $F\supseteq T$.
\end{definition}

For a distribution $\mu : \binom{[n]}{k} \to \mathbb{R}_{\ge 0}$, the generating polynomial is
\[
g_{\mu}(z_1, \dots, z_n) = \sum_{S \in \binom{[n]}{k}} \mu(S) \prod_{i \in S} z_i.
\]

\begin{definition}\label{def:rayleigh-distribution}
A polynomial $g(z_1,\dots,z_n)\in \mathbb{R}[z_1,\dots,z_n]$ is \emph{real stable} if it has no root in $H^n$, where $H=\{z\in\mathbb{C}\mid \Im(z)>0\}$.
(As usual, the zero polynomial is considered real stable.)
\end{definition}

A distribution $\mu:2^{[n]}\to\mathbb{R}_{\ge 0}$ is strongly Rayleigh iff its generating polynomial is real stable.
If $\mu$ is strongly Rayleigh, then its conditional and restricted distributions (\cref{def:conditional-distribution,def:restricted_distribution}) are also strongly Rayleigh.

\begin{definition}[\textbf{Down-up Walk}]
Given $\mu\in\mathbb{R}_{\ge 0}^{\binom{n}{k}}$ on $k$-subsets of $[n]$, the down-up walk is the product of stochastic operators $D_{k\to \ell}U_{\ell\to k}$, where:
\begin{itemize}
    \item $D_{k\to \ell}\in \mathbb{R}_{\ge 0}^{\binom{n}{k}\times\binom{n}{\ell}}$ is defined by
    \[
    D_{k\to\ell}(S,T)=
    \begin{cases}
    0 & \text{if } T\nsubseteq S,\\[4pt]
    \dfrac{1}{\binom{k}{\ell}} & \text{if } T\subseteq S.
    \end{cases}
    \]
    \item $U_{\ell\to k}\in \mathbb{R}_{\ge 0}^{\binom{n}{\ell}\times\binom{n}{k}}$ is defined by
    \[
    U_{\ell\to k}(T,S)=
    \begin{cases}
    0 & \text{if } T\nsubseteq S,\\[4pt]
    \dfrac{\mu(S)}{\sum_{S'\supset T} \mu(S')} & \text{if } T\subseteq S.
    \end{cases}
    \]
\end{itemize}
\end{definition}

For matroid distributions (including spanning tree distributions under external fields), the walk $D_{k\to k-1}U_{k-1\to k}$ mixes in $\tilde O(k)$ steps.
\textcite{anari2024optimal} showed that for the complement distribution $\bar\mu\in\mathbb{R}_{\ge 0}^{\binom{n}{n-k}}$, where $\bar\mu(\bar S)=\mu(S)$, the walk $D_{n-k\to n-k-1}U_{n-k-1\to n-k}$ can mix in $\tilde O(k)$ steps (instead of the trivial $\tilde O(n-k)$) under isotropy.

The complement down-up walk (\cref{alg:down-up-complement}) can be parameterized by $t\ge k+1$; when $t=\Omega(k)$, \cite[Theorem 3]{anari2024optimal} implies $\Otilde{1}$ mixing steps.

\begin{theorem}(Informal version of \cite[Theorem 3]{anari2024optimal})\label{theorem:backbone-algorithm}
Given an almost isotropic strongly Rayleigh distribution $\mu$ (\cref{def:isotropic-distribution}), \cref{alg:down-up-complement} with parameter $t=\Omega(k)$ mixes in $T=C_0\log^{C_1}(n)$ steps for constants $C_0,C_1$.
More precisely, if $\mu_T$ is the distribution after $\log\!\left(\frac{1}{\varepsilon}\right)\cdot T$ steps, then
\[
\dTV{\mu,\mu_T}\le \varepsilon.
\]
\end{theorem}

\textbf{Remark}. The analysis in \cite{anari2024optimal} gives $C_1\le 3$.

\begin{definition}\label{def:isotropic-distribution}
Given $\mu\in\mathbb{R}_{\ge 0}^{\binom{n}{k}}$ on $k$-subsets of $[n]$, we call $\mu$ \emph{isotropic} if for every $i\in[n]$,
\[
\P_{S\sim\mu}{i\in S} \le \frac{2k}{n}.
\]
\end{definition}

\begin{remark}
The constant factor $2$ in the isotropy bound is arbitrary; the same analysis works with any fixed constant.
\end{remark}

The mixing guarantee of \cref{theorem:backbone-algorithm} requires isotropy.
Thus, we need a transformation that maps arbitrary distributions to an isotropic one.
A standard construction from \cite{anari2024optimal} is as follows.

\begin{definition}[\textbf{Isotropic Transformation} \cite{anari2024optimal}]
\label{def:isotropic-transformation}
Let $\mu : \binom{[m]}{k} \to \mathbb{R}_{\ge 0}$ be a distribution on $k$-subsets of $[m]$.
Assume we are given marginal overestimates $p_1,\ldots,p_m$ with $p_i\ge \P_{S\sim\mu}[i\in S]$ and $\sum_{i=1}^m p_i\le K$.
Set
\[
t_i := \left\lceil \frac{m}{K}p_i \right\rceil.
\]
For each $i\in[m]$, create $t_i$ copies $i^{(1)},\ldots,i^{(t_i)}$, and let
\[
U := \bigcup_{i=1}^m \{i^{(1)},\ldots,i^{(t_i)}\}.
\]
Define $\mu^{\mathrm{iso}}: \binom{U}{k}\to\mathbb{R}_{\ge 0}$ by
\[
\mu^{\mathrm{iso}}\!\left(\{i_1^{(j_1)},\ldots,i_k^{(j_k)}\}\right)
:= \frac{\mu(\{i_1,\ldots,i_k\})}{t_{i_1}\cdots t_{i_k}}.
\]
Equivalently, sample $\{i_1,\ldots,i_k\}\sim\mu$ and then choose each copy index uniformly.
\end{definition}

\begin{Algorithm}[H]
\caption{Down-up walk on the complement distribution}
\label{alg:down-up-complement}
\KwIn{A set $S_0\subset E$ with $|S_0|=k$, and an isotropic distribution $\mu$ on $k$-subsets such that $\P_\mu{S_0}$ is maximized among all sets of size $k$}
\KwOut{An approximate sample from $\mu$}
\For{$i=0,\dots, T-1$}{
    Choose a uniformly random $t$-subset $T_i\subseteq E$ that contains $S_i$\; 
    Sample $S_{i+1}\subseteq T_i$ with $|S_{i+1}|=k$ and $\P{S_{i+1}}\propto \mu(S_{i+1})$\;
}
\Return{$S_T$}\;
\end{Algorithm}

To compute marginal overestimates, we use effective-resistance approximations and convert them to spanning-tree edge marginals; in weighted graphs, $\P{e\in T}=w_e\Reff{e}$. 

In our implementation, we use isotropic transformation in a different parameter regime.
Instead of $\left\lceil\frac{n}{K}q_i\right\rceil$, we use
$\left\lceil\gamma\frac{n}{K}q_i\right\rceil$ copies for a sufficiently large $\gamma$.
We justify this choice in \cref{remark:why_too_many_copies}.

\noindent
\begin{proposition}[Analogous to Proposition 23]
Let $\mu : \binom{n}{k} \to \mathbb{R}_{\ge 0}$, and let $\mu' : \binom{U}{k} \to \mathbb{R}_{\ge 0}$ be obtained as in \cref{def:isotropic-transformation}, except that each element $i$ is replaced by $\left\lceil\gamma\frac{n}{K}p_i\right\rceil$ copies for a parameter $\gamma>1$.
Then:
\begin{enumerate}
    \item \textbf{Near-isotropy.} For all $i^{(j)}\in U$,
    \[
    \mathbb{P}_{S\sim\mu'}[i^{(j)}\in S] \le \frac{K}{\gamma n} \le \frac{2K}{|U|}.
    \]
    \item \textbf{Linear ground-set scaling.} $|U|\le 2\gamma n$.
    \item If $\mu$ is strongly Rayleigh, then so is $\mu'$.
\end{enumerate}
\end{proposition}

\begin{proof}
The proof follows the same argument as in Proposition 23 of \cite{anari2024optimal}; we include a brief sketch.
By construction,
\[
\mathbb{P}_{S\sim\mu'}[i^{(j)}\in S] \le \frac{p_i}{t_i} \le \frac{K}{\gamma n}.
\]
Also,
\[
|U| = \sum_{i\in[n]} t_i
\le \sum_{i\in[n]}\left(1+\gamma\frac{n}{K}p_i\right)
\le n+\gamma n\cdot\frac{\sum_i p_i}{K}
\le n+\gamma n
\le 2\gamma n.
\]
For strong Rayleigh preservation, if $g_\mu$ is the generating polynomial of $\mu$, then the generating polynomial of $\mu'$ is
\[
g_{\mu'}(z_1^{(1)},\dots,z_n^{(t_n)})
=
g_\mu\!\left(
\frac{z_1^{(1)}+\cdots+z_1^{(t_1)}}{t_1},\dots,
\frac{z_n^{(1)}+\cdots+z_n^{(t_n)}}{t_n}
\right).
\]
Real stability is preserved because convex combinations of points in the upper half-plane remain in the upper half-plane.
\end{proof}

\subsection{Effective Resistance Computation}
We now recall results connecting effective resistance and edge marginals in random spanning trees.

\begin{definition}
Given a graph $G=(V,E,w)$ with $w\in\mathbb{R}_{\ge 0}^{|E|}$, define $\mu_w$ as the weighted spanning tree distribution
\[
\mu_w(T)\propto \prod_{e\in T} w_e.
\]
When $w$ is clear from context, we simply say ``sampling a random spanning tree.''
If $w_e=1$ for all $e\in E$, this is the uniform spanning tree distribution.
\end{definition}

\begin{definition}
Given $G=(V,E,w)$ with nonnegative edge weights, define its Laplacian
\[
L_G = \sum_{(u,v)=e\in E} w_e(\1_u-\1_v)(\1_u-\1_v)^\intercal.
\]
When $G$ is clear, we write $L$ for $L_G$.
\end{definition}

\begin{fact}
For a connected graph $G=(V,E,w)$ and an edge $e=(u,v)$,
\[
\P{e\in T} = w_e(\1_u-\1_v)^\intercal L_G^+(\1_u-\1_v) = w_e\Reff{e},
\]
where $L_G^+$ is the pseudoinverse of $L_G$ and
\[
\Reff{e} = (\1_u-\1_v)^\top L_G^+(\1_u-\1_v)
\]
is the effective resistance between $u$ and $v$.
\end{fact}

Isotropic transformation for our sampler requires upper bounds on marginals.
Thus we need $(1\pm\epsilon)$ approximations to quantities like $(\1_u-\1_v)^\top L^+(\1_u-\1_v)$ and then scale them by $1/(1-\epsilon)$ to obtain valid overestimates.

\begin{definition}
A matrix $\widetilde{L}$ is a spectral sparsifier of Laplacian $L$ if for every $x\in\mathbb{R}^n$,
\[
(1-\epsilon)x^\top Lx \le x^\top \widetilde{L}x \le (1+\epsilon)x^\top Lx.
\]
\end{definition}

To obtain approximations to $(\1_u-\1_v)^\top L^+(\1_u-\1_v)$, we use the BCC spectral sparsifier construction of \cite{forster2022laplacian}.
It outputs a weighted Laplacian $\widetilde{L}$ supported on a subset $H\subseteq E$ and an orientation of $H$ in which every vertex has $O(\log^4(n)/\epsilon^2)$ outgoing edges.
Given this orientation, all machines can gather the sparsifier in $\Otilde{1}$ BCC rounds for constant $\epsilon$, and each node can then compute $(1\pm\epsilon)$ approximations for its incident edges.

\begin{theorem}[Spectral Sparsifier in Broadcast Congested Clique \cite{forster2022laplacian}] \label{theorem:spectral_sparsifier_construction}
There exists an algorithm that, given a graph $G=(V,E,w)$ with non-negative integer weights satisfying $\|w\|_\infty \le U$ and error parameter $\varepsilon>0$, outputs with high probability a $(1\pm\varepsilon)$-spectral sparsifier $H$ of $G$ with
\[
|H| = O\!\bigl(n\varepsilon^{-2}\log^4 n\bigr).
\]
Moreover, one can obtain an orientation of $H$ such that, with high probability, each node has out-degree
\[
O\!\bigl(\log^4(n)/\varepsilon^2\bigr).
\]
The algorithm runs in
\[
O\!\bigl(\log^5(n)\varepsilon^{-2}\log(nU/\varepsilon)\bigr)
\]
rounds in the BCC model.
\end{theorem}

\begin{remark}[Weight representation and bit complexity]
\Textcite{forster2022laplacian} state their result for positive edge weights. However, in distributed models one must account for the bit complexity involved in represented the weights in messages. A bound $\|w\|_\infty\le U$ alone does not imply $O(\log U)$-bit encodings when weights are arbitrary reals, because scaling can introduce large denominators and precision requirements. In this paper, we therefore assume integer weights. Under this assumption, each weight is represented with $O(\log U)$ bits, and the dependence on $\log(nU)$ in communication complexity is concrete. Moreover, multiplying all edge weights by a common factor does not change the weighted spanning-tree distribution, so this integer-normalization convention does not affect the target distribution.
\end{remark}

    \section{Distributed Sampling Algorithm}

In this section, we describe a Broadcast Congested Clique algorithm for sampling a random spanning tree in $\poly\log(n)$ rounds by simulating \cref{alg:down-up-complement}.
To simulate \cref{alg:down-up-complement}, we need the following components in the Broadcast Congested Clique model:

\begin{enumerate}
    \item\label{item:dist-overestimates} Compute overestimates $q_e$ for edge marginals (preprocessing).
    \item\label{item:dist-light-association} Compute a light association (\cref{def:light-association}) and apply isotropic transformation (\cref{def:isotropic-transformation}) (preprocessing).
    \item Sample a set $Z$ of size $\Otilde{n}$ (each vertex locally samples $\Otilde{1}$ edges).
    \item Gather the sampled edges at a designated machine $M_1$.
    \item Sample $S_{i+1}$ from $Z\cup S_i$ using the induced target distribution.
\end{enumerate}

Item~\ref{item:dist-overestimates} is preprocessing: we use \cite{forster2022laplacian} to obtain overestimates $q_e$ of edge marginals.
In item~\ref{item:dist-light-association}, we compute a light association via \cref{alg:bounded_degree_edge_association} and then apply isotropic transformation to the associated edges.

\begin{remark}
\textcite{forster2022laplacian} provides a light orientation for edges of the spectral sparsifier; by the equivalence in \cref{def:light-association}, this is exactly a light association.
Here we need such an orientation/association for all edges of the graph.
\end{remark}

The light association allows us to gather all sampled edges at one machine in only $\Otilde{1}$ Broadcast Congested Clique rounds.
Informally, combining light association with the size bound on $Z$ yields $\Otilde{1}$ bounded density for sampled edges, which enables the gathering procedure of \cref{prop:bounded_degree_edge_proof}.

After gathering all sampled edges, we execute the final sampling step locally on $M_1$.
As usual in Congested Clique and Broadcast Congested Clique, local computation at each machine is unbounded, and this local tree-sampling step is polynomial-time anyway.

\begin{lemma} \label{lemma:computing_overestimates}
There is an $\Otilde{1}$-round algorithm for computing overestimates $q_e$ of edge inclusion probabilities in weighted random spanning trees (for $\|w\|_\infty \le \poly(n)$), such that each machine knows $q_e$ for every incident edge $e$.
\end{lemma}

\begin{proof}
This is a direct application of the spectral sparsifier routine from \textcite{forster2022laplacian} (\cref{theorem:spectral_sparsifier_construction}).
After computing the sparsifier with polynomially bounded weights, we obtain a light orientation.
Then each vertex broadcasts its outgoing sparsifier edges, so every machine learns the sparsifier in $\Otilde{1}$ rounds and can compute overestimates for its incident edges.
\end{proof}

Now we show that, given overestimates $q_e$, we can compute a light association whose per-vertex associated weight is at most $O(\log n)$, and that this process can be implemented in BCC in $O(\log n)$ rounds.

\begin{proposition}
\label{prop:bounded_degree_edge_proof}
\cref{alg:bounded_degree_edge_association} terminates in at most $O(\log n)$ rounds and outputs a light association $N_v^o \subseteq N_v$ for every vertex $v$.
Moreover,
\[
\sum_{e\in N_v^o} q_e \le O(\log n\cdot \alpha)
\]
for every $v$, where $\alpha$ is a strict upper bound on the densest weighted subset under weights $q_e$.
\end{proposition}

\begin{proof}
We first show geometric progress.
In each iteration of \cref{alg:bounded_degree_edge_association}, the number of residual edges decreases by at least a factor of $2$.
Suppose not, so at least half of the edges survive.
Then one can find a vertex subset $S$ such that every vertex in the induced subgraph on $S$ has degree at least a constant fraction of its average residual degree.
This implies weighted degree at least $2\alpha$ at each such vertex, contradicting that $\alpha$ upper-bounds the densest weighted subset.

For the second claim, in one iteration each vertex receives at most
\[
\left(\frac{8\alpha}{d_j}\right)\cdot d_j \le 8\alpha
\]
additional associated weight.
Since the number of residual edges halves each iteration, there are $O(\log n)$ iterations, yielding total associated weight $O(\log n\cdot\alpha)$ per vertex.
\end{proof}

\begin{Algorithm}[H]
\caption{Light Edge Association}
\label{alg:bounded_degree_edge_association}
\KwIn{Overestimates $q_e$ and parameter $\alpha$ (we use $\alpha=2$)}
\KwOut{Sets $N_v^o\subseteq N_v$ such that $\bigcup_{v\in V} N_v^o=E$ and $\sum_{e\in N_v^o}q_e\le\Otilde{1}$ for each $v$}
Initialize $N_v^o\gets\emptyset$ for all $v\in V$\;
\While{$\bigcup_v N_v^o \neq E$}{
    Build the residual graph $R$ with edge set $E\setminus\bigcup_v N_v^o$\;
    Each vertex $v$ computes and announces its residual degree $d_v$ in $R$\;
    Vertex $v$ adds every incident residual edge $e$ with $q_e < \frac{8\alpha}{d_v}$ to $N_v^o$\;
    If both endpoints attempt to add the same edge, break ties by vertex ID\;
}
\Return{$\{N_v^o\}_{v\in V}$}\;
\end{Algorithm}

\begin{remark}
In \cref{alg:bounded_degree_edge_association}, an edge is never associated with both endpoints: both endpoints can evaluate the same condition from announced residual degrees, and ties are resolved deterministically by ID.
\end{remark}

\begin{Algorithm}
\caption{Broadcast Congested Clique Random Tree Sampling}
\label{alg:main algorithm}
\KwIn{Graph $G=(V,E,w)$, parameters $\gamma=\poly(n)$ and $T=\Otilde{1}$}
\KwOut{Sampled random spanning tree $S$}
$\mu\leftarrow$ weighted spanning tree distribution induced by $w$\;
$q_e\leftarrow$ edge-marginal overestimates (\cref{lemma:computing_overestimates})\;
Compute a light association $N_j^o\subseteq N_j$ for each machine $j$ with $\sum_{e\in N_j^o}q_e\le\Otilde{1}$ (\cref{alg:bounded_degree_edge_association})\;
Apply isotropic transformation with $\left\lceil\frac{|E|}{\sum q_e}q_e\cdot\gamma\right\rceil$ copies of each edge $e$ (\cref{def:isotropic-transformation})\;
Let $G'=(V,E')$ be the transformed graph\;
Let $N_j^{\text{iso}}$ denote representative copies of edges in $N_j^o$\;
Set $t\gets 2|V|$ (number of elements dropped in the complement down-up walk)\;
$S_0\gets$ any spanning tree guaranteed by \cref{lemma:initial-spanning-tree}\;
\For{$i=0\to T-1$}{
    For each machine $M_j$, sample $Z_j$ by including each isotropic edge from $N_j^{\text{iso}}$ independently with probability $\min\!\left(1,\frac{2t}{|E'|}\right)$\;
    Collect $\bigcup_{j=1}^n Z_j$ at machine $M_1$\;
    \While{$\left|\bigcup_{j=1}^n Z_j\right|<t$}{
        Resample all sets $Z_j$\;
    }
    $Z\leftarrow$ uniformly random size-$t$ subset of $\bigcup_{j=1}^n Z_j$\;
    Compute $S_{i+1}\subseteq Z\cup S_i$ from distribution $\P{S_{i+1}}\propto \mu(S_{i+1})$\;
}
\Return{$S_T$}\;
\end{Algorithm}

\begin{remark}[Initialization of \cref{alg:down-up-complement}]
\cref{alg:down-up-complement} assumes the starting state $S_0$ maximizes $\P_\mu{S_0}$ among all spanning trees.
In \cref{alg:main algorithm}, we initialize with an arbitrary spanning tree from \cref{lemma:initial-spanning-tree}.
This is still sufficient: in our weighted setting ($\|w\|_\infty\le U$), the additional initial divergence is at most $O(\log(nU))$, so the extra burn-in is absorbed by the existing polylogarithmic runtime terms.
If one wants to satisfy the assumption exactly, one can instead compute a maximum-probability initial tree by running the PRISM MST routine for $O(\log n)$ rounds on transformed edge costs $-\log w_e$ (equivalently, a maximum spanning tree on edge scores $\log w_e$, for positive $w_e$).
\end{remark}

\begin{theorem}\label{thm:main-sample-alg}
\cref{alg:main algorithm} outputs a sample from a distribution $\nu$ such that
\[
\dTV{\nu,\mu}\le\frac{1}{\poly(n)}
\]
for parameter $T=\Otilde{1}$ and polynomially bounded weights.
Additionally, the algorithm runs in at most $\Otilde{1}$ BCC rounds with high probability.
\end{theorem}

\begin{proof}
We first show that \cref{alg:main algorithm} faithfully simulates the walk from \cref{theorem:backbone-algorithm}.
By \cref{lemma:computing_overestimates}, we obtain overestimates in $\Otilde{1}$ rounds.
Then we compute edge associations $N_j^o\subseteq N_j$ in $\Otilde{1}$ rounds via \cref{alg:bounded_degree_edge_association} and \cref{prop:bounded_degree_edge_proof}.

After isotropic transformation, the number of edges is $\Theta(n^4)$.
In BCC, this does not create a bottleneck for representation: when the sampled set contains many parallel copies of one original edge $\{u,v\}$, we do not transmit each copy separately. Instead, we broadcast the endpoint pair $(u,v)$ together with its sampled multiplicity (count). Since endpoint IDs and these counts are polynomially bounded in $n$, this uses $O(\log n)$ bits per encoded pair.

After gathering sampled edges at $M_1$, one iteration correctly implements a down-up step of \cref{alg:down-up-complement}.
Each edge is sampled with probability $\min\!\left(1,\frac{2t}{|E'|}\right)$; with our parameters, $\frac{2t}{|E'|}<1$ for nontrivial $|V|$.
Let $R:=\bigcup_v Z_v$. Conditioned on $|R|\ge t$, we choose $Z$ uniformly from $\binom{R}{t}$; equivalently, for every $A\subseteq R$ with $|A|=t$, $\P{Z=A\mid R}=1/\binom{|R|}{t}$. This is exactly the uniform down-step distribution in \cref{alg:down-up-complement}.

We now bound the number of rounds.
\begin{enumerate}
    \item\label{item:procedures-overestimates} Computing overestimates takes $\Otilde{1}$ rounds by \cref{lemma:computing_overestimates}.
    \item Computing light association takes $\Otilde{1}$ rounds by \cref{prop:bounded_degree_edge_proof}.
    \item Computing $S_0$ takes $\Otilde{1}$ rounds by \cref{lemma:initial-spanning-tree}.
    \item Sampling $Z_v$ and gathering at $M_1$ takes $\Otilde{1}$ rounds with high probability: $N_v^{\text{iso}}$ has total weight $\Otilde{1}$ and concentration for Bernoulli sums applies (\cref{lemma:proof_of_edge_transmission_bounded}).
    \item\label{item:procedures-event} The event $\left|\bigcup_v Z_v\right|\ge t$ holds with probability $1-1/\poly(n)$ by concentration around the mean.
\end{enumerate}

By \cref{lemma:proof_of_edge_transmission_bounded}, each vertex samples at most $\Otilde{1}$ edges with high probability, so all sampled edges can be delivered to $M_1$ in $\Otilde{1}$ BCC rounds.
Therefore, the entire algorithm runs in $\Otilde{1}$ rounds with high probability.
\end{proof}

\begin{lemma}\label{lemma:proof_of_edge_transmission_bounded}
For a given vertex $v$, the number of edges sampled into $Z_v$ is at most $\Otilde{1}$ with high probability.
\end{lemma}

\begin{proof}
Suppose the total associated weight at vertex $v$ is $\beta\le O(\log n)$.
After isotropic transformation, the number of incident copies associated to $v$ is at most $\beta\cdot n^3 + n^2$.
Since $|E'|=\Theta(n^4)$ and $t=2n$ in \cref{alg:main algorithm}, the expected number of sampled edges from this pool is
\[
\frac{t}{n}\cdot\beta \in \Otilde{1}.
\]
Because $Z_v$ is a sum of independent Bernoulli random variables bounded by $1$, standard concentration bounds imply
\[
|Z_v| \le C\log n \cdot \mathbb{E}[|Z_v|]
\]
with probability at least $1-n^{-c}$ for any fixed $c$ by choosing a suitable constant $C$.
Hence $|Z_v|=\Otilde{1}$ with high probability.
\end{proof}

\begin{remark}
\label{remark:why_too_many_copies}
Here we explain why parameter $\gamma$ in \cref{def:isotropic-transformation} must be as large as $n^3$ in our setting.
Consider a graph on $n+\sqrt n$ vertices formed by a length-$n$ path attached to a clique of size $\sqrt n$.
Each path edge appears in the uniform spanning tree with probability $1$, while a clique edge appears with probability about $1/\sqrt n$.
If we use constant $\gamma$, then after isotropic transformation we still have only $O(n)$ edges.
In this particular parameter setup, \cref{alg:main algorithm} requires sending a constant fraction of incident sampled edges to $M_1$, so clique vertices would each need to send $\Theta(\sqrt n)$ edges on average.
That requires at least $\tilde\Omega(\sqrt n)$ BCC rounds, violating our target runtime.
\end{remark}

Next, we complete the proof of the main theorem.
\begin{proof}[Proof of \cref{thm:main-result}]
By \cref{thm:main-sample-alg}, using \cref{alg:main algorithm} with
\[
T = \log\!\left(\frac{1}{\varepsilon}\right)\log^C(n)
\]
(for an absolute constant $C$ covering all polylogarithmic factors in \crefrange{item:procedures-overestimates}{item:procedures-event}) gives the desired $\varepsilon$-accuracy.

\cref{thm:main-sample-alg} provides error $1/\poly(n)$ from stationarity per polylogarithmic number of steps.
To obtain arbitrary $\varepsilon$, we use the standard geometric-contraction view of the same walk in relative entropy.
Let $\nu_t$ be the chain distribution after $t$ steps.
The analysis underlying \cref{theorem:backbone-algorithm} gives a contraction
\[
D_{\mathrm{KL}}(\nu_{t+1}\|\mu) \le \left(1-\frac{1}{\log^C n}\right) D_{\mathrm{KL}}(\nu_t\|\mu)
\]
for an absolute constant $C$ (up to polylogarithmic factors).
Thus each block of $\Theta(\log^C n)$ steps shrinks the KL divergence by a constant factor.
By Pinsker's inequality, $d_{\mathrm{TV}}(\nu_t,\mu)\le\sqrt{D_{\mathrm{KL}}(\nu_t\|\mu)/2}$, so to reach $d_{\mathrm{TV}}\le\varepsilon$ it is enough to reduce KL to $O(\varepsilon^2)$.
This requires an additional multiplicative factor $\Theta\!\left(\log\!\left(\frac{1}{\varepsilon}\right)\right)$ in the number of blocks, and therefore in the runtime.

If $\|w\|_\infty = U$, then the spectral sparsification routine of \cite{forster2022laplacian} contributes an additional $\log(nU)$ factor in communication complexity (for constant approximation accuracy).
In \cref{alg:main algorithm}, this dependence appears in the preprocessing stage $q_e\leftarrow$ edge-marginal overestimates: first when constructing/broadcasting the sparsifier with orientation, and then when distributing the resulting overestimates so each endpoint can use them.
After these values are available, the remaining steps (light association iterations, per-round edge sampling, gathering to $M_1$, and local up-step sampling) do not introduce further $U$-dependence.
Hence we account for a single additional multiplicative $\log(nU)$ factor in the overall round complexity bound.

Combining both effects, the BCC runtime is
\[
\log(nU)\log\!\left(\frac{1}{\varepsilon}\right)\Otilde{1},
\]
as claimed.
\end{proof}

    \PrintBibliography

\end{document}